\documentclass[a4paper,UKenglish,cleveref, autoref, thm-restate]{lipics-v2021}

\usepackage{style}
\usepackage{shortcuts}

\hideLIPIcs
\title{Node-Weighted Triangles: Faster and Simpler}
\author{Shyan Akmal}{Max Planck Institute for Informatics, Saarbr\"{u}cken, Germany  \and \url{https://www.shyanakmal.com/} }{sakmal@mpi-inf.mpg.de}{https://orcid.org/0000-0002-7266-2041}{This work received funding from the Klaus Tschira Boost Fund, a joint initiative of GSO Guidance, Skills \& Opportunities e.V. and the Klaus Tschira Stiftung.}
\author{Nick Fischer}{Max Planck Institute for Informatics, Saarbr\"{u}cken, Germany \and \url{https://nick-fischer.com/} }{nfischer@mpi-inf.mpg.de}{https://orcid.org/0009-0001-0909-3296}{}
\authorrunning{S. Akmal and N. Fischer}
\Copyright{Shyan Akmal and Nick Fischer}


\category{} %optional, e.g. invited paper
\relatedversion{} %optional, e.g. full version hosted on arXiv, HAL, or other respository/website
\acknowledgements{We thank the anonymous reviewers for their helpful comments. }

\nolinenumbers %uncomment to disable line numbering

	\theoremstyle{remark}
	\newtheorem{assumption}[theorem]{Assumption}

\newcommand{\nwtri}{\textup{\textsf{Node-Weighted Triangle}}}
\newcommand{\nwt}{\textup{\textsf{NWT}}}
\newcommand{\minnwtri}{\textup{\textsf{Minimum Node-Weighted Triangle}}}
\newcommand{\minnwt}{\textup{\textsf{Min-NWT}}}

\newcommand{\weight}{\textup{\textsf{wt}}}
\newcommand{\Tri}{\textup{\textsf{Tri}}}

\EventEditors{Sayan Bhattacharya, Danupon Nanongkai, Michael Benedikt, and Gabriele Puppis}
\EventNoEds{4}
\EventLongTitle{53rd International Colloquium on Automata, Languages, and Programming (ICALP 2026)}
\EventShortTitle{ICALP 2026}
\EventAcronym{ICALP}
\EventYear{2026}
\EventDate{July 7--10, 2026}
\EventLocation{Royal Holloway, University of London, Egham, United Kingdom}
\EventLogo{}
\SeriesVolume{374}
\ArticleNo{75}

\begin{document}

\maketitle

\begin{abstract}
    Weighted variants of triangle detection are an important object of study because of their prominence in fine-grained complexity. We revisit the \nwtri{} problem, where the goal is to decide if a vertex-weighted graph contains a triangle whose node weights sum to zero. 
    This problem has been the focus of a celebrated line of work, beginning with a subcubic-time algorithm [Vassilevska, Williams; STOC~'06], and culminating in algorithms running \emph{almost} in matrix multiplication time, \smash{$O(\mmult(n) + n^2\cdot 2^{O(\sqrt{\log n})})$} [Czumaj, Lingas; SODA~'07], [Vassilevska~W., Williams; STOC~'09].
    This runtime is \emph{almost-optimal}, since even detecting an unweighted triangle is conjectured to require matrix multiplication time $\mmult(n)$. 
    However, the superpolylogarithmic~\smash{$2^{\Omega(\sqrt{\log n})}$} overhead persists in a world where near-optimal matrix multiplication is possible (i.e., \makebox{$\mmult(n) \leq n^2\poly(\log n)$}).
    
    In this paper, we present a new algorithm solving \nwtri{} in $O(\mmult(n))$ time,  closing the gap to unweighted triangle detection completely. 
    Remarkably, our algorithm is much \emph{simpler} than previous approaches, which use involved recursion schemes and communication protocols.
    \end{abstract}

\section{Introduction}

    What makes some problems hard, and others easier? 
    Fine-grained complexity attempts to answer this question by identifying a small collection of primitive hard problems that explain the intractability of large classes of computational tasks. 
    One set of core hard problems that has proven particularly influential in this context is triangle detection and its variants.

    Triangle detection is the problem of finding a \emph{triangle} (i.e., a set of three mutually adjacent vertices) in an $n$-node graph.
    The simple brute force algorithm solves this problem in $O(n^3)$ time.
    It turns out however that we can detect triangles much faster in $O(\mmult(n))$ time, where $\mmult(n) \le O(n^{2.3716})$ is the time complexity of multiplying~$n\times n$ matrices \cite{AlmanDWXXZ25}. This algorithm is conjectured to be optimal, i.e., it is believed that triangle detection requires at least $(\mmult(n))^{1-o(1)}$ time to solve~\cite[Hypothesis 6]{V2019}. 
    % \nick{What's the right reference? Maybe just folklore and we omit the citation?}

    More challenging variants of triangle detection include \emph{weighted} versions. For instance, the task of detecting a triangle with minimum total weight in an edge-weighted graph can still be solved in $O(n^3)$ time by brute force, but no polynomial improvement over this runtime is known. In fact, this problem turns out to be \emph{equivalent} to the fundamental \textsf{All-Pairs Shortest Paths} (\textsf{APSP}) problem~\cite{VW2018}, and beating the simple $O(n^3)$-time algorithm would refute the \textsf{APSP Hypothesis}, a cornerstone of fine-grained complexity.
    %
    % An important variant of triangle detection is \emph{edge-weighted} triangle detection,
    % where  we are given a graph with edge weights and are tasked with finding a triangle whose sum of edge weights is minimized.
    % This edge-weighted triangle detection problem can still be solved in $O(n^3)$ time by brute force,
    % but no polynomial improvement over this runtime is known.
    % In particular, it is open whether matrix multiplication can be used to find edge-weighted triangles faster.
    %
    % It is conjectured that unweighted triangle detection and edge-weighted triangle detection require $(\mmult(n))^{1-o(1)}$ time and $n^{3-o(1)}$ time to solve respectively, so that existing algorithms are essentially optimal \cite{V2019}.
    % These hardness hypotheses have important implications for the complexity of many central problems in graph algorithms.
    %
    % For example, the edge-weighted triangle detection problem can be solved in truly subcubic time if and only if we can find shortest paths between all pairs of vertices in a weighted, directed graph in truly subcubic time \cite{VW2018}. 
    % Similarly, an ``all-edges'' version of unweighted triangle detection can be solved in faster than $O(\mmult(n))$ time if and only if we can compute the transitive closure of a directed graph in faster than $O(\mmult(n))$ time. 
    More broadly, researchers have established a vast array of graph  problems whose complexity status (e.g., ``requires cubic time,'' ``solvable in matrix-multiplication time,'' or ``has an intermediate runtime strictly between matrix-multiplication and cubic time'') is governed by the complexity of special variants of triangle detection \cite{VW2013,VW2018,VX2020,LPV2020,XC2024}.
    These connections make understanding the complexity of generalizations of triangle detection, especially weighted versions, important. 
    
    In this paper we study the \nwtri{} (\nwt{}) problem, arguably the simplest weighted generalization of triangle detection.
    In this problem, we are given a vertex-weighted graph $G$,
    and are tasked with determining if $G$ has a triangle whose node weights sum to zero.
    We also consider the \minnwtri{} (\minnwt{}) problem, where we are given the same input, but are now tasked with finding a triangle whose sum of node weights is minimized. 
    
    As before, \nwt{} and \minnwt{} can be solved in $O(n^3)$ time by brute force. But is this the best possible? What is the true fine-grained complexity of \nwt{} and \minnwt{}? This was the driving question of a short but influential line of work. In~2006, Vassilevska and Williams~\cite{VW2006} presented the first truly subcubic algorithms, running in $O(\sqrt{n^3\cdot \mmult(n)}\log n)$ time. This runtime was improved slightly in~\cite{VWY2006} using fast rectangular matrix multiplication.
    Shortly after, Czumaj and Lingas~\cite{CL2009} devised a completely different approach based on a clever recursion scheme to achieve runtime \smash{$O(\mmult(n) + n^2\cdot 2^{O(\log n/\log\log n)})$}. Subsequently, Vassilevska W.\ and Williams~\cite{VW2013} optimized this  approach to solve \nwt{} and \minnwt{} in 
    \smash{$O(\mmult(n) + n^2\cdot 2^{O(\sqrt{\log n})})$} time.
    Alternate methods solving \nwt{} in \smash{$\mmult(n)\cdot2^{O(\sqrt{\log n})}$} and \smash{$\mmult(n)\cdot 2^{O(\log n/\log\log n)}$} time were also presented in 
    \cite{VW2013} and \cite{ALW2014}, based on efficient multiparty communication protocols and vector-based weight reduction arguments respectively.
    
    In summary, this line of work successfully settles that \nwt{} has the same complexity as plain triangle detection up to $n^{o(1)}$ factors. However, the presence of these superpolylogarithmic factors
    is unsatisfactory given how basic a question \nwt{} is.
    Even if we were to design near-optimal matrix multiplication algorithms in the future to achieve $\mmult(n) \le n^2\poly(\log n)$,
    due to this overhead the existing algorithms for \nwt{} would all still require at least \smash{$n^2 \cdot 2^{\Omega(\sqrt{\log n})}$} time.
    This overhead is an artifact of the complicated arguments used in previous approaches (e.g., intricate recursive arguments, or communication protocols based off constructions of large sets avoiding arithmetic progressions).
        
    In this paper, we present a simpler, faster algorithm for solving \nwtri{}.

    \begin{restatable}{theorem}{exact}
        \label{thm:exact}
        \nwtri{} can be solved in  $O(\mmult(n))$ time.
    \end{restatable}

    In particular, we show that \nwt{} is as easy to solve as unweighted triangle detection,
    with no overhead in the asymptotic runtime whatsoever.
    Perhaps our most significant contribution is that we design a truly simple algorithm.
    Unlike previous work,
    it does not rely on sophisticated recursion schemes or communication protocols.
    Instead, we simply run an unweighted triangle detection algorithm on a small collection of subgraphs of the input (chosen based off the relative frequencies of vertex weights in the graph) and argue that this directly solves the problem.
    Intuitively, our argument works by reducing \nwt{} to the structured version of \nwt{} where each vertex weight appears the same number of times. 
    This ``weight-regular'' version of \nwt{} turns out to be much easier to solve.

    By a standard reduction from \minnwt{} to \nwt{} \cite[Theorem 3.3]{VW2013}, we immediately get the following result for \minnwt{} as well.

    \begin{corollary}
            \minnwtri{}
            can be solved in $O(\mmult(n)\cdot \log W)$ time on graphs whose vertex weights have absolute value at most $W$.
    \end{corollary}

    \begin{remark}
        In the literature, the running time of $n\times n$ matrix multiplication is typically written as $O(n^\omega)$,
        where $\omega$ is the ``exponent of matrix multiplication,''
        the infimum over all reals $\tau\ge 2$ such that $n\times n$ matrices can be multiplied in $O(n^\tau)$ time.
        The use of the infimum in this definition means that, technically, we can compute the product of two $n\times n$ matrices in $O(n^{\omega+\varepsilon})$ time for any constant $\varepsilon > 0$, 
        but not necessarily in $O(n^\omega)$ time.
        Since our work involves optimizing subpolynomial factors in the time complexity of \nwt{},
        we denote the runtime of matrix multiplication by $\mmult(n)$ to avoid confusion.
    \end{remark}

    \subsection{Additional Related Work}
    Vertex-weighted versions of graph problems beyond triangle detection have received more attention in the last few years.
    For example, \cite{AFJVX2025}  proved that \textsf{APSP} with node weights can be solved in $\sqrt{n^3\cdot \mmult(n)}\poly(\log n)$ time, even though the classic edge-weighted version of \textsf{APSP} is conjectured to require $n^{3-o(1)}$ time \cite{V2019}.
    The classic \textsf{Minimum Cut} problem can be solved in near-linear time \cite{K2000}, but extending this result to the node-capacitated case remains open \cite{CT2025}. 
    Researchers have also studied node-weighted versions of maximum matching \cite{AP2022}.
    % and dominating set \cite{ALW2014}.

    % \subsection{Organization}

\section{Preliminaries}

    Throughout we let $G = (V,E)$ denote the input graph on $n$ vertices.
    All graphs we consider have vertex weights, with $\weight(v)$ denoting the weight of vertex $v$. 

    We let $\mmult(a,b,c)$ denote the time complexity of multiplying an $a\times b$ matrix with a $b\times c$ matrix. Let $\mmult(n) = \mmult(n, n, n)$. 
    We make use of the standard fact that an unweighted triangle in a tripartite graph with vertex parts $X, Y, Z$ can be detected in $\mmult(|X|, |Y|, |Z|)$ time (see e.g., the argument presented in \cite[Section 5.3]{IR1978}). 
    
    We make the following  mild assumption on $\mmult(n)$:

    \begin{assumption}
        \label{assumption:monotone}
        $\mmult(n)/n^2$ is nondecreasing in $n$.\footnote{\Cref{assumption:monotone} intuitively asserts that the runtime of matrix multiplication should be monotonically increasing in the dimensions of the matrices (as trivially $\mmult(n)\ge n^2$).
        This is a standard type of assumption that has been used in previous work on weighted triangle problems \cite{VW2018}. It holds for reasonable runtimes such as those of the form \smash{$\mmult(n) =  n^\alpha(\log n)^\beta$} for constants $\alpha, \beta\ge 0$, and is mostly for notational convenience: unconditionally, we obtain the more verbose statement that \nwt{} can be solved in time $O(\sum_i \mmult(n_i))$ for parameters $n_i$ satisfying $\sum_i n_i^2 \leq O(n^2)$.}
    \end{assumption}

\section{Node-Weighted Triangle Detection}
    The intuition behind our algorithm is simple. Suppose that the graph contains~$g$ distinct weights, each appearing $\Theta(n / g)$ times. 
    In this simplified scenario, one can easily solve \nwt{} as follows. 
    Enumerate all pairs of possible weights $(w_1, w_2)$. If there is a solution $(x, y, z)$ with $\weight(x) = w_1$ and $\weight(y) = w_2$, then the weight of the third node must be $\weight(z) = -(w_1 + w_2)$. 
    To detect such a solution it suffices to find an \emph{unweighted} triangle spanning the subsets~$X, Y, Z$ defined by
    \begin{align*}
        X &= \set{x \in V : \weight(x) = w_1}, \\
        Y &= \set{y \in V : \weight(y) = w_2}, \\
        Z &= \set{z \in V : \weight(z) = -(w_1 + w_2)}.
    \end{align*}
    This can be done with matrix multiplication in $\mmult(|X|, |Y|, |Z|) \le O(\mmult(n / g))$ time, so the total runtime is $O(g^2 \cdot \mmult(n/g)) \le O(\mmult(n))$ by \Cref{assumption:monotone}. 
    Our algorithm works by generalizing this basic idea to the setting where not all weights have the same frequency.

    \begin{lemma} \label{lem:one-weight-fix}
        Let $H$ be tripartite with vertex parts $X, Y, Z$
        such that \emph{every} node in $X$ has the same weight. Then we can solve \nwt{} on $H$ in $\mmult(|X|, |Y|, |Z|)$ time.
    \end{lemma}
    \begin{proof}
        Let $w$ be the unique vertex weight in $X$. Remove all edges $(y, z) \in Y \times Z$ with $w + \weight(y) + \weight(z) \neq 0$. The triangles in the remaining graph are exactly the triangles whose node weights sum to zero. 
        So to solve \nwt{} on $H$, it suffices to detect an unweighted triangle in the remaining graph in $\mmult(|X|, |Y|, |Z|)$ time.
    \end{proof}

    With this simple lemma, we can now prove our main result.

    \exact*
    \begin{proof}
        Let $W = \set{\weight(x) : x \in V}$ be the set of node weights in $G$. 
        For any weight $w \in W$, let $f(w) = \abs{\set{v \in V : \weight(v) = w}}$ denote its frequency. The set $W$ and all frequencies can be computed in $O(n)$ time by scanning through the vertices of the graph. 
        Given this data,
        we run the following algorithm:

        \medskip
        \begin{algorithmic}[1]
            \ForEach{$w \in W$} \label{step:overall}
                \State Let $W' = \set{w' \in W : f(w') \leq f(w)}$
                \State \label{step:greedy}Let $\mathcal P$ be a greedy partition of $W'$ into parts $P$ with $\sum_{w' \in P} f(w') \leq 2 f(w)$
                \ForEach{$P \in \mathcal P$} \label{alg:main:line:loop-inner}
                    \State Let $X = \set{x \in V : \weight(x) = w}$
                    \State Let $Y = \set{y \in V : \weight(y) \in P}$
                    \State Let $Z = \set{z \in V : -(w + \weight(z)) \in P}$
                    \State Detect a solution in $X \times Y \times Z$ by \cref{lem:one-weight-fix} \label{alg:main:line:detect}
                \EndForEach
            \EndForEach
        \end{algorithmic}
        \vspace{0.1cm}
        
        To elaborate on line \ref{step:greedy} above: we scan through the set $W'$ in some fixed order, and keep adding weights to our current part. 
        Right before the sum of the frequencies of weights in our current part would exceed $2f(w)$, we set aside that part and start growing a new part.
        Once we have completed scanning through $W$ in this way, we have produced the partition $\mathcal{P}$.
        Since all weights in $W'$ have frequency at most $f(w)$ and all frequencies are positive, each part we build is guaranteed to include at least one element, so this process terminates.

        \proofsubparagraph*{Correctness.}
        We claim that if $G$ has a solution $(x, y, z)$, then we will detect it in line \ref{alg:main:line:detect}. 
        Without loss of generality, suppose $f(x) \geq f(y)$. 
        Consider the iteration of the algorithm where we pick $w = \weight(x)$ in line \ref{step:overall}. 
        Then $\weight(y) \in W'$ by definition, hence there is a part $P$ in the partition from line \ref{step:greedy} with $\weight(y) \in P$. By construction $x \in X$ and $y \in Y$.
        Since $(x,y,z)$ is a solution, we must have $z\in Z$.
        Thus we successfully detect the solution $(x,y,z)$ in line \ref{alg:main:line:detect}.

        \proofsubparagraph*{Running Time.}
        Fix an iteration of the algorithm where we pick weight $w$ in line \ref{step:overall}. 
        We compute ~$W'$ and~$\mathcal P$ in $O(n)$ time by scanning through $W$. 
        Consider the loop in lines \ref{alg:main:line:loop-inner}-\ref{alg:main:line:detect}. 
        Let $X_P, Y_P, Z_P$ be the sets $X, Y, Z$ constructed in the iteration of this inner loop when we pick part $P$ in line \ref{alg:main:line:loop-inner}. 
        By construction $|X_P| = f(w)$ and $|Y_P| = \sum_{w' \in P} f(w') \leq 2f(w)$. Since $\mathcal P$ partitions a subset $W'$ of weights in $G$, we have $\sum_{P \in \mathcal P} |Z_P| \leq n$. 
        Moreover, as each weight in $W'$ has frequency at most $f(w)$ and we form a new part in $\mathcal{P}$ only when the total frequency would exceed $2f(w)$, each part $P\in\mathcal{P}$ has total frequency at least $f(w)$, so  the partition has at most $|\mathcal{P}| \le n/f(w)$  parts.
        Each call to \cref{lem:one-weight-fix} takes $O(\mmult(|X_P|, |Y_P|, |Z_P|))$ time, so the total time spent in the loop is
        \begin{align*}
            O\parens*{\sum_{P \in \mathcal P} \mmult(|X_P|, |Y_P|, |Z_P|)}
            &\le O\parens*{\sum_{P \in \mathcal P} \mmult(f(w), f(w), |Z_P|)} \\
            &\le O\parens*{\sum_{P \in \mathcal P} \grp{\frac{|Z_P|}{f(w)} + 1} \cdot \mmult(f(w))} \\
            &\le O\parens*{\frac{n}{f(w)} \cdot \mmult(f(w))} \\
            &\le O\parens*{f(w) \cdot \frac{\mmult(n)}{n}}.
        \end{align*}
        The transition from the first to second line follows from the fact that in order to multiply an $f(w)\times f(w)$ matrix with an $f(w)\times |Z_P|$ matrix, if $|Z_P| > f(w)$ we can split the second matrix into the concatenation of $O(|Z_P|/f(w))$ many $f(w)\times O(f(w))$ matrices,
        and simply multiply the first $f(w)\times f(w)$ matrix with each of the new $f(w)\times O(f(w))$ matrices separately.
        If $|Z_P| \le f(w)$, then we can directly multiply the matrices in at most $\mmult(f(w))$ time.
        The transition from the second to third lines follows from $\sum_{P\in\mathcal{P}} |Z_P|\le n$ and $|\mathcal{P}| \le n/f(w)$.
        The transition from the penultimate to 
        final line follows from \cref{assumption:monotone}. 
         
        Summing over all possible choices of $w\in W$ in line \ref{step:overall}, we get that the algorithm takes
        \begin{equation*}
            O\parens*{\sum_{w \in W} \parens*{n + f(w) \cdot \frac{\mmult(n)}{n}}} \le O\parens*{n^2 + n \cdot \frac{\mmult(n)}{n}} \le O(n^2 + \mmult(n)) \le O(\mmult(n))
        \end{equation*}
        time overall as claimed.
    \end{proof}

\section{Additional Remarks}

    We conclude by observing some simple extensions of our argument, and placing them in the context of previous work on node-weighted pattern detection.

\subparagraph*{Reduction to Unweighted Triangle Detection.}

    The only place we use matrix multiplication in our algorithm for \nwt{} is in calls to \Cref{lem:one-weight-fix}.
    In the proof of \Cref{lem:one-weight-fix}, we only use matrix multiplication to detect unweighted triangles.
    This means that our algorithm for \nwt{} can be interpreted as an efficient \emph{reduction} to unweighted triangle detection.
    Concretely, if we 
    let $\Tri(n)$ denote the time complexity of detecting a triangle in an $n$-node graph,
    our arguments actually prove that \nwt{} can be solved in $O(\Tri(n))$ time.
    This is relevant in case it turns out that unweighted triangles actually can be detected in faster than matrix multiplication time.
    The previous algorithms of \cite{CL2009,VW2013} can also be formulated as reductions to unweighted triangle detection.

\subparagraph*{Real Weights.}

    If instead of integer weights the graph has real node weights, our algorithm still applies unchanged. 
    This is because we never use the fact that the weights are integers, only that we can add and compare weights in constant time. 
    Previous algorithms for \nwt{} handle this case as well.

\subparagraph*{Counting Weighted Triangles.}

    The algorithms of \cite{VW2013} can solve the \emph{counting} version of \nwt{}, where we are tasked with returning the \emph{number} of triangles $(x,y,z)$ whose sum of node weights is zero. 
    Our approach can be adapted to solve this counting problem faster in $O(\mmult(n))$ time as well, with just two small changes.

    The first change is to replace \Cref{lem:one-weight-fix} with an algorithm that counts rather than just detects node-weighted triangles when one part has uniform weight.
    This follows from the standard fact that we can count unweighted triangles using matrix multiplication.
    In our algorithm we then keep a running count of solutions.
    
    The second change is necessary to avoid overcounting triangles. Here we distinguish three types of triangles: (1) all three nodes have distinct weights, (2) exactly two nodes share the same weight, (3) all three nodes share the same weight. To count triangles of type~(1), we first artificially ensure that the frequencies $f(w)$ of all weights are distinct. Then, by inspecting the correctness argument of \Cref{thm:exact}, we see that a solution $(x,y,z)$ with distinct node weights will be counted three times by the algorithm (twice when the weight in the solution with highest frequency is selected in line \ref{step:overall} of the algorithm, and once when the weight in the solution with second-highest frequency is selected). 
    The arguments for (2) and (3) are only easier. For instance, to count the triangles of type (2) we enumerate all weights~$w$ of the first two nodes; the third node necessarily has weight $-2w$. Then we count the unweighted triangles in $X_w\times X_w\times Z_{w}$, where $X_w$ is the set of nodes with weight $w$ and $Z_w$ is the set of nodes with weight $-2w$. Each type-(2) triangle is counted exactly twice, and the total runtime is $\sum_w \mmult(|X_w|, |X_w|, |Z_{w}|) \le O(\mmult(n))$ by a similar analysis as in \Cref{thm:exact}. 
    % \nick{I shortened this a bit; are you still okay with the level of detail? Feel free to revert.}

    % To count solutions where exactly two node weights are equal,
    % we loop over all nonzero weights $w\in W$,
    % and count the number of solutions where two nodes have weight $w$.
    % In any such solution, the third node has weight $-2w$.
    % For each choice of $w$ let $X_w$ be the set of nodes with weight $w$ and $Z_{w}$ be the set of nodes with weight $-2w$.
    % We can
    % count unweighted triangles in $X_w\times X_w\times Z_{w}$ in $\mmult(|X_w|, |X_w|, |Z_{w}|)$ time using matrix multiplication.
    % By definition, $|X_w| = f(w)\le n$.
    % Since the $Z_{w}$ sets are disjoint for distinct $w$, we have $\sum_{w\in W} |Z_{w}| \le n$.
    % Thus the analysis from the proof of \Cref{thm:exact} shows we can compute all of these matrix products in at most $O(\mmult(n))$ time.

    % Combining these counts appropriately solves the problem as desired.

\subparagraph*{Larger Subgraph Patterns.}

    Suppose that instead of seeking a triangle whose sum of node weights is zero, we are tasked with finding a subgraph isomorphic to a fixed pattern graph $H$ on $k$ vertices whose sum of node weights is zero.
    Then combining \Cref{thm:exact} with standard reductions (see e.g., the discussion at the beginning of \cite[Section 3]{VW2013}) shows that when $3\mid k$, we can solve this more general node-weighted $H$ detection problem in  $O(\mmult(n^{k/3}))$ time. 
        If instead $k\equiv 1\pmod 3$ or $k\equiv 2\pmod 3$,
    we have an additional $n$ or $n^2$ overhead on top of an $O(\mmult(n^{\lfloor k/3\rfloor}))$ runtime respectively.

\subparagraph*{Sparse Graphs.}
    It is well-known that  we can find an unweighted triangle in an $m$-edge graph in $O(m^{2\omega/(\omega+1)})$ time, provided $\mmult(n) \le O(n^\omega)$ \cite[Theorem 3.5]{AYZ1997}. 
    This is proved by reducing triangle detection in sparse graphs to triangle detection in dense graphs with fewer vertices.
    As observed in \cite[Proof of Corollary 5]{CL2009}, this argument generalizes to the \nwt{} problem.
    In particular, if $\mmult(n) \le O(n^\omega)$, previous work solves \nwt{} in $m$-edge graphs in $O(m^{2\omega/(\omega+1)})$ time for $\omega > 2$, but for $\omega = 2$  only achieves a runtime of \smash{$m^{4/3}\cdot 2^{O(\sqrt{\log m})}$}~\cite{VW2013}.
    With our new algorithm for \nwt{}, we can remove the \smash{$2^{\Theta(\sqrt{\log m})}$}~factor in the runtime, and find node-weighted triangles as quickly as unweighted triangles in sparse graphs, for the full range of possible values of $\mmult(n)$.
\bibliography{main}

\end{document}